\documentclass[11pt]{article}

\usepackage{amsmath,amssymb,amsthm,color}
\usepackage{fullpage,natbib,authblk}
\usepackage[boxed]{algorithm2e}

\title{Sparse Solutions to Nonnegative Linear Systems and Applications}
\author[1]{Aditya Bhaskara\thanks{Email: \textsf{bhaskaraaditya@google.com}}}
\author[2]{Ananda Theertha Suresh\thanks{Email: \textsf{asuresh@ucsd.edu}}}
\author[1]{Morteza Zadimoghaddam\thanks{Email: \textsf{zadim@google.com}}}

\affil[1]{Google Research NYC}
\affil[2]{University of California, San Diego}

\date{}

\newcommand{\norm}[1]{\left \lVert #1\right\rVert }

\newcommand{\rr}{\mathbb{R}}
\newcommand{\rplus}{\rr_{+}}
\newcommand{\eps}{\epsilon}

\newcommand{\E}{\mathbb{E}}
\newcommand{\xopt}{x^{*}}
\newcommand{\xalg}{x_{\text{alg}}}
\newcommand{\su}[1]{^{(#1)}}
\newcommand{\ytil}{\widetilde{y}}
\newcommand{\ones}{\mathbf{1}}
\newcommand{\iprod}[1]{\langle #1 \rangle}
\newcommand{\calN}{\mathcal{N}}
\newcommand{\tr}[1]{\text{Tr}\left( #1 \right)}

\newtheorem{theorem}{Theorem}[section]
\newtheorem{lemma}[theorem]{Lemma}
\newtheorem{defn}[theorem]{Definition}

\newtheorem{remark}[theorem]{Remark}

\newcommand{\tcO}{\tilde{\mathcal{O}}}
\newcommand{\cO}{\mathcal{O}}
\newcommand{\cS}{\mathcal{S}}
\newcommand{\cI}{\mathcal{I}}

\newcommand{\mupi}{\mups r}

\newcommand{\backskip}{\vspace{-0.05in}}

\newcommand{\bp}{{p}}
\newcommand{\bx}{{x}}
\newcommand{\by}{{y}}
\newcommand{\bv}{{v}}
\newcommand{\bfvec}{{f}}
\newcommand{\mup}{\mu}
\newcommand{\mups}[1]{\mu_{#1}}
\newcommand{\sigp}{{\sigma}}
\newcommand{\sigps}[1]{{\sigma}_{#1}}

\newcommand{\sigpi}{\sigps r}

\newcommand{\wi}{w_r}
\newcommand{\wj}{w_r}

\begin{document}
\maketitle

\vspace{-.3in}
\begin{abstract}
We give an efficient algorithm for finding sparse approximate solutions to linear systems of 
equations with nonnegative coefficients. Unlike most known results for sparse recovery, we do not require {\em any}
assumption on the matrix other than non-negativity. Our algorithm is combinatorial in nature, 
inspired by techniques for the {\em set cover} problem, as well as the multiplicative weight update method. 

We then present a natural application to learning mixture models in the PAC framework.
For learning a mixture of $k$ axis-aligned Gaussians in $d$ dimensions, we give an algorithm that 
outputs a mixture of $O(k/\epsilon^3)$ Gaussians that is $\epsilon$-close in statistical distance 
to the true distribution, without any separation assumptions. The time and sample complexity is 
roughly $O(kd/\epsilon^3)^{d}$. This is polynomial when $d$ is constant -- precisely the regime
in which known methods fail to identify the components efficiently.

Given that non-negativity is a natural assumption, we believe that our result may find use in other settings in which we wish to approximately {\em explain} data using a small number of a (large) candidate set of components.
\end{abstract}

\section{Introduction}\label{sec:intro}
{\em Sparse recovery}, or the problem of finding sparse
solutions (i.e., solutions with a few non-zero entries) to linear systems of equations, is a fundamental problem in 
signal processing, 
machine learning and theoretical computer science. In its simplest form, 
the goal is to find a solution to a given system of equations $Ax = b$ that minimizes $\norm{x}_0$ 
(which we call the {\em sparsity} of $x$).

It is known that sparse recovery is NP hard in general. 
It is related to the question of finding if a set of points in $d$-dimensional space are in {\em general position} --
i.e., they do not lie in any $(d-1)$ dimensional subspace~\citep{Khachiyan}.
A strong negative result in the same vein is due to~\cite{Arora93} and (independently)~\cite{Kann93},
who prove that it is not possible to approximate the quantity 
$\min \{ \norm{x}_0 : Ax = b\}$ to a factor better than $2^{(\log n)^{1/2}}$ unless NP 
has quasi polynomial time algorithms. 

While these negative results seem forbidding,
there are some instances in which sparse recovery is possible. 
Sparse recovery is a basic problem in the field of compressed sensing, and in a beautiful line of work, 
\cite{CandesTao, Donoho} and others show that convex relaxations can be used for sparse recovery when 
the matrix $A$ has certain structural properties, such as {\em incoherence},
or the so-called restricted isometry property (RIP). However the focus in compressed sensing is to  {\em design} matrices $A$ (with as few rows or `measurements' as possible) 
that allow the recovery of sparse vectors $x$ given $Ax$. 
Our focus is instead on solving the sparse recovery problem for a {\em given} $A, b$, 
similar to that of \citep{Natarajan95,DonohoM03}. 
In general, checking if a given $A$ possesses the RIP is a hard problem~\citep{RIPCertify}. 

Motivated by the problem of PAC learning mixture models (see below), we consider the sparse recovery problem when the matrix $A$, the vector $b$, and the solution we seek all have {\em non-negative entries}. In this case, we prove that {\em approximate} sparse recovery is always possible, with some loss in the sparsity. We obtain the following trade-off:
\begin{theorem}\label{thm:main-informal}(Informal)
Suppose the matrix $A$ and vector $b$ have non-negative entries, and suppose there exists a $k$-sparse\footnote{I.e., has at most $k$ nonzero entries.} non-negative $\xopt$ such that $A \xopt = b$. Then for any $\eps >0$, there is an efficient algorithm that produces an $\xalg$ that is $O(k/\eps^3)$ sparse, and satisfies $\norm{ A \xalg - b}_1 \le \eps \norm{b}_1$.
\end{theorem}

The key point is that our upper bound on the error is in the $\ell_1$ norm (which is the largest among all $\ell_p$ norms).
Indeed the trade-off between the sparsity of the obtained solution and the error is much better understood if the error is measured in the $\ell_2$ norm. In this case, the natural greedy `coordinate ascent', as well as the algorithm based on sampling from a ``dense'' solution give non-trivial guarantees (see~\cite{Natarajan95, Shalev-Shwartz}). If the columns of $A$ are normalized to be of unit length, and we seek a solution $x$ with $\norm{x}_1 = 1$, one can find an $x'$ such that $\norm{Ax'-b}_2 < \eps$ and $x'$ has only $O(\frac{\log (1/\eps)}{\eps^2})$ non-zero co-ordinates. A similar bound can be obtained for general convex optimization problems, under strong convexity assumptions on the loss function~\citep{Shalev-Shwartz}.

While these methods are powerful, they do not apply (to the best of our knowledge) when the error is measured in the $\ell_1$ norm, as in our applications. More importantly, they do not take advantage of the fact that there {\em exists} a $k$-sparse solution (without losing a factor that depends on the largest eigenvalue of $A^{\dagger}$ as in~\cite{Natarajan95}, or without additional RIP style assumptions as in~\cite{Shalev-Shwartz}). 

The second property of our result is that we do not rely on the {\em uniqueness} of the solution (as is the case with approaches based on convex optimization). Our algorithm is more combinatorial in nature, and is inspired by multiplicative weight update based algorithms for the {\em set cover} problem, as described in Section~\ref{sec:algo}.
Finally, we remark that we do not need to assume that there is an ``exact'' sparse solution (i.e., $A\xopt = b$), and a weaker condition suffices. See Theorem~\ref{thm:main-approx} for the formal statement. 

Are there natural settings for the sparse recovery problem with non-negative $A, b$? 
One application we now describe is that of learning mixture models in the PAC framework~\citep{Valiant, KearnsMRRSS94}.

\subsubsection*{Learning mixture models}
A common way to model data in learning applications is to view it as
arising from a ``mixture model'' with a small number of parameters.
Finding the parameters often leads to a better understanding of the
data. The paradigm has been applied with a lot of success to data in
speech, document classification, and so on~\citep{ReynoldsR95,TitteringtonSM85,Lindsay95}.
Learning algorithms for Gaussian mixtures, hidden Markov models, topic models for documents, etc. have received wide attention both in theory and practice.

In this paper, we consider the problem of learning a mixture of Gaussians. Formally, given samples from a mixture of $k$ Gaussians in $d$ dimensions, the goal is to recover the components with high probability.  The problem is extremely well studied, starting with the early heuristic methods such as expectation-maximization (EM). The celebrated result of~\cite{Dasgupta99} gave the first rigorous algorithm to recover mixture components, albeit under a {\em separation} assumption. This was then improved in several subsequent works (c.f. \cite{SanjeevK01, VempalaW02, DasguptaS07}). 


More recently, by a novel use of the classical method of moments,~\cite{KalaiMV10} and~\cite{BelkinS10} showed that any 
$d$-dimensional Gaussian mixture with a constant number of components $k$ can be recovered in polynomial time (without any strong separation). 
However the dependence on $k$ in these works is exponential. 
\cite{MoitraV10} showed that this is {\em necessary} if we wish to recover the {\em true} components, even in one dimension.


In a rather surprising direction,~\cite{HsuK13}, and later \cite{BhaskaraCMV14} and \cite{BelkinMore}
showed that if the dimension $d$ is large (at least $k^{c}$ for a constant $c>0$), then 
tensor methods yield polynomial time algorithms for parameter recovery, under mild non-degeneracy assumptions. Thus  the case of small $d$ and much larger $k$ seems to be the most challenging for current techniques, if we do not have separation assumptions. Due to the lower bound mentioned above, we cannot hope to recover the true parameters used to generate the samples.

\backskip
\paragraph{Our parameter setting.}
We consider the case of constant $d$, and arbitrary $k$. As mentioned earlier, this case has sample complexity exponential in $k$ if we wish to recover the {\em true} components of the mixture (\cite{MoitraV10}). We thus consider the corresponding PAC learning question (\cite{Valiant}): given parameters $\eps, \delta>0$ and samples from a mixture of Gaussians as above, can we find a mixture of $k$ Gaussians such that the statistical distance to the original mixture is $< \eps$ with success probability (over samples) $\ge (1-\delta)$?

\backskip
\paragraph{Proper vs improper learning.}
The question stated above is usually referred to as {\em proper} learning: given samples from a distribution $f$ in a certain class (in this case a mixture of $k$ Gaussians), we are required to output a distribution $\hat{\bfvec}$ in the same class, such that $\norm{ \bfvec - \hat{\bfvec} }_1 \leq \eps$.  A weaker notion that is often studied is {\em improper} learning, in which $\hat{\bfvec}$ is allowed to be arbitrary (it some contexts, it is referred to as {\em density estimation}).

Proper learning is often much harder than improper learning. To wit, the best known algorithms for proper learning of Gaussian mixtures run in time exponential in $k$. It was first studied by~\cite{FeldmanSO06}, who gave an algorithm with sample complexity polynomial in $k, d$, but run time exponential in $k$. Later works improved the sample complexity, culminating in the works by \cite{DaskalakisK14, AcharyaJOS14}, who gave algorithms with {\em optimal} sample complexity, for the case of spherical Gaussians.
We note that even here, the run times are $\text{poly}(d,1/\eps)^k$.

Meanwhile for improper learning, there are efficient algorithms known for 
learning mixtures of very general one dimensional distributions 
(monotone, unimodal, log-concave, and so on). A sequence of works by \cite{ChanDSS13, ChanDSS14} give algorithms that have near-optimal sample complexity (of $\tcO(k/\eps^2)$), and run in polynomial time. However it is not known how well these methods extend to higher dimensions.

In this paper we consider something in between proper and improper learning. We wish to return a mixture of Gaussians, but with one relaxation: we allow
the algorithm to output a mixture with slightly more than $k$ components. Specifically, we obtain a tradeoff between the number of components in the output mixture, and the distance to the original mixture. Our theorem here is as follows

\begin{theorem}\label{thm:gauss-informal}(Informal)
Suppose we are given samples from a mixture of $k$ axis-aligned Gaussians in $d$ dimensions. There is an algorithm with running time and sample complexity $O \left( \frac{1}{\eps^3} \cdot \big( \frac{kd}{\eps^3} \big)^{d} \right)$, and outputs a mixture of $O(k/\eps^3)$ axis-aligned Gaussians which is $\eps$-close in statistical distance to the original mixture, with high probability.
\end{theorem}

The algorithm is an application of our result on solving linear systems. Intuitively, we consider a matrix whose columns are the probability density functions (p.d.f.) of all possible Gaussians in $\rr^d$, and try to write the p.d.f. of the given mixture as a sparse linear combination of these. To obtain finite bounds, we require careful discretization, which is described in Section~\ref{sec:gaussians}.

\paragraph{Is the trade-off optimal?}
It is natural to ask if our tradeoff in Theorem~\ref{thm:main-informal} is the best possible (from the point of view of efficient algorithms). We conjecture that the optimal tradeoff is $k/\eps^2$, up to factors of $O(\log (1/\eps))$ in general. We can prove a weaker result, that for obtaining an $\epsilon$ approximation in the $\ell_1$ norm to the general sparse recovery problem using polynomial time algorithms, we cannot always get a sparsity better than $k\log (1/\eps)$ unless $\mathcal{P} = \mathcal{NP}$. 

While this says that {\em some} dependence on $\eps$ is necessary, it is quite far from our algorithmic bound of $O(k/\eps^3)$.  In Section~\ref{sec:lower-bounds}, we will connect this to similar disparities that exist in our understanding of the set cover problem. We present a {\em random planted} version of the set cover problem, which is beyond all known algorithmic techniques, but for which there are no known complexity lower bounds. We show that unless this planted set cover problem can be solved efficiently, we cannot hope to obtain an $\eps$-approximate solution with sparsity $o(k/\eps^2)$. This suggests that doing better than $k/\eps^2$ requires significantly new algorithmic techniques.


\subsection{Basic notation}\label{sec:notation}
We will write $\rplus$ for the set of non-negative reals. 
For a vector $x$, its $i$th co-ordinate will be denoted by $x_i$, and for a matrix $A$, $A_i$
denotes the $i$th column of $A$. For vectors $x, y$, we write $x \le y$ to mean entry-wise inequality. 
We use $[n]$ to denote the set of integers $\{1, 2, \dots, n\}$. 
For two distributions $p$ and $q$, we use $\norm{p-q}_1$ to denote the $\ell_1$ distance between them.

\section{Approximate sparse solutions}\label{sec:algo}
\subsection{Outline}
Our algorithm is inspired by techniques for the well-known {\em set cover}
problem: given a collection of $n$ sets $S_1,
S_2, \dots, S_n \subseteq [m]$, find the sub-collection of the
smallest size that covers all the elements of $[m]$. In our problem,
if we set $A_i$ to be the indicator vector of the set $S_i$, and $b$
to be the vector with all entries equal to one, a sparse solution to
$Ax=b$ essentially covers all the elements of $[m]$ using only a {\em few}
sets, which is precisely the set cover problem. The difference between the two problems is that in linear equations, we are required to `cover' all the elements
{\em precisely once} (in order to have equality), and additionally, we are allowed to use sets fractionally.

Motivated by this connection, we define a potential function which captures the
notion of covering all the elements ``equally''. For a vector
$x \in \rr^n$, we define
\begin{equation}\label{eq:phi-defn}
\Phi(x) := \sum_{j} b_j (1+\delta)^{(Ax)_j/b_j}
\end{equation}

This is a mild modification of the potential function used in the multiplicative weight update method~(\cite{FreundS97, AroraHK12}).
Suppose for a moment that $\norm{b}_1 =1$. Now, consider
some $x$ with $\norm{x}_1 = 1$. If $(Ax)_j = b_j$ for all $j$, the
potential $\Phi(x)$ would be precisely $(1+\delta)$. On the other
hand, if we had $(Ax)_j/b_j$ varying significantly for different $j$,
the potential would (intuitively) be significantly larger; this
suggests an algorithm that tries to increment $x$ coordinate-wise,
while keeping the potential small. Since we change $x$
coordinate-wise, having a small number of iterations implies
sparsity. The key to the analysis is to prove that at any point in the
algorithm, there is a ``good'' choice of coordinate that we can
increment so as to make progress. We now make this intuition formal, and prove the following

\begin{theorem}\label{thm:main-approx}
Let $A$ be an $m \times n$ non-negative matrix, and $b \in \rr^m$ be a non-negative vector. Suppose there exists a $k$-sparse non-negative vector $\xopt$ such that $\norm{A\xopt}_1 = \norm{b}_1$ and $A\xopt \le (1+\eps_0) b$, for some $0< \eps_0 < 1/16$. Then for any $\eps \ge 16\eps_0$, there is an efficient algorithm that produces an $\xalg$ that is $O(k/\eps^3)$ sparse, and satisfies $\norm{ A \xalg - b}_1 \le \eps \norm{b}_1$.
\end{theorem}

\subsubsection*{Normalization}
For the rest of the section, $m, n$ will denote the dimensions of $A$, as in the statement of Theorem~\ref{thm:main-approx}. Next, note that by scaling all the entries of $A, b$ appropriately, we can assume without loss of generality that $\norm{b}_1 = 1$.  Furthermore, since for any $i$, multiplying $x_i$ by $c$ while scaling $A_i$ by $(1/c)$ maintains a solution, we may assume that for all $i$, we have $\norm{A_i}_1=1$ (if $A_i=0$ to start with, we can simply drop that column). Once we make this normalization, since $A, b$ are non-negative, any non-negative solution to $Ax = b$ must also satisfy $\norm{x}_1 = 1$.

\subsection{Algorithm}
We follow the outline above, having a total of $O(k/\eps^3)$
iterations. At iteration $t$, we maintain a solution $x \su{t}$, obtained
by incrementing precisely one co-ordinate of $x\su{t-1}$. We start
with $x\su{0} = 0$; thus the final solution is $O(k/\eps^3)$-sparse.

We will denote $y\su{t} := (Ax\su{t})$. Apart from the potential
$\Phi$ introduced above (Eq.\eqref{eq:phi-defn}), we keep track of
another quantity:
\[ \psi(x) := \sum_j (Ax)_j. \]
Note that since the entries of $A, x$ are non-negative, this is simply
$\norm{Ax}_1$.  
\begin{algorithm}[!h]
  ~~procedure \textsf{solve}($\{A, b, k, \eps\}$) \\
  ~~\texttt{// find an
  $\eps$-approximate solution. } \BlankLine \Begin{

      \nl Initialize $x\su{0} = 0$; set parameters $T = Ck/\delta^2$; $C = 16/\eps$; $\delta = \eps/16$. \label{alg:step-init}

      \For {$t = 0, \dots, T-1$} { \nl Find a coordinate $i$ and a
          scaling $\theta >0$ such that $\theta \ge 1/Ck$, and the
          ratio $\Phi(x\su{t} + \theta e_i)/\Phi(x\su{t})$ is minimized.
          
          \nl Set $x\su{t+1} = x\su{t} + \theta e_i$.

      }

      \nl Output $\xalg = x\su{t}/\norm{x\su{t}}_1$.
    }
\end{algorithm}

\paragraph{Running time.} Each iteration of the algorithm can be easily implemented in time $O( mn \log (mn)/\delta )$ by going through all the indices, and for each index, checking for a $\theta$ in multiples of $(1+\delta)$.

Note that the algorithm increases $\psi(x\su{t})$ by at least $1/Ck$ in every iteration (because the increase is precisely $\theta$, which is $\ge 1/Ck$), while increasing $\Phi$ as slowly as possible. Our next lemma says that once $\psi$ is large enough (while having a good bound on $\Phi$), we can get a ``good'' solution. I.e., it connects the quantities $\Phi$ and $\psi$ to the $\ell_1$ approximation we want to obtain.

\begin{lemma}\label{lem:l1-bound}
Let $x \in \rr^n$ satisfy the condition $\Phi(x) \le (1+\delta)^{(1+\eta) \psi(x)}$, for some $\eta >0$. Then we have 
\begin{equation}
\norm{ \frac{(Ax)}{\psi(x)} - b}_1 \le 2\left( \eta + \frac{1}{\delta \psi(x)} \right).
\end{equation}
\end{lemma}
\begin{proof}
For convenience, let us write $y = Ax$, and $\ytil = \frac{y}{\norm{y}_1}$ (i.e., the normalized version). Note that since each column of $A$ has unit $\ell_1$ norm, we have $\psi(x) = \norm{Ax}_1 = \norm{y}_1$. Since $\ytil$ and $b$ are both normalized, we have
\[ \norm{\ytil - b}_1 = 2 \cdot \sum_{j ~:~ \ytil_j > b_j} (\ytil_j - b_j). \]
From now on, we will denote $S := \{j~:~ \ytil_j >b_j\}$, and write $p := \sum_{j \in S} b_j$. Thus to prove the lemma, it suffices to show that
\begin{equation}\label{eq:suffices1}
\sum_{j \in S} (\ytil_j - b_j) \le \left( \eta + \frac{1}{\delta \psi(x)} \right).
\end{equation}
Now, note that the LHS above can be written as $\sum_{j \in S} b_j \big( \frac{\ytil_j}{b_j} - 1 \big)$. We then have
\begin{align*}
(1+ &\delta)^{\frac{\psi(x)}{p} \cdot \sum_{j \in S} b_j \big( \frac{\ytil_j}{b_j} -1 \big) } \le (1+\delta)^{\sum_{j \in S} \frac{b_j}{p} \cdot \big( \frac{y_j}{b_j} - \psi(x) \big)} \\
& \le \sum_{j \in S} \frac{b_j}{p} \cdot (1+\delta)^{\big( \frac{y_j}{b_j} -\psi(x) \big)} \quad \text{(convexity)} \\
& \le \frac{1}{p} \cdot \sum_{j} b_j (1+\delta)^{ \big( \frac{y_j}{b_j} - \psi(x) \big) } \quad \text{ (sum over {\em all} $j$)} \\
& \le \frac{1}{p} \cdot \Phi(x) \cdot (1+\delta)^{-\psi(x)} \\
& \le \frac{1}{p} \cdot (1+\delta)^{ \eta \psi(x) } \qquad \text{(hypothesis on $\Phi$).}
\end{align*}
Thus taking logarithms (to base $(1+\delta)$), we can bound the LHS of Eq.\eqref{eq:suffices1} by
\[ \frac{p}{\psi(x)} \log_{(1+\delta)} (1/p) + p \eta \le \frac{1}{\delta \psi(x)} + \eta. \]
The last inequalities are by using the standard facts that $\ln (1+\delta) \ge \delta/2$ (for $\delta<1$), and $p \ln (1/p) \le (1/e)$ for any $p$, and since $p\le 1$.
This shows Eq.~\eqref{eq:suffices1}, thus proving the lemma.
\end{proof}

The next lemma shows that we can always find an index $i$ and a $\theta$ as we seek in the algorithm.

\begin{lemma}\label{lem:progress}
Suppose there exists a $k$-sparse vector $x^*$ such that $\Phi(x^*) \le (1+\delta)^{(1+\eps_0)}$. Then for any $C>1$, and any $x\su{t} \in \rr^n$, there exists an index $i$, and a scalar $\theta \ge 1/(Ck)$, such that
\[ \Phi( x\su{t} + \theta e_i ) \le (1+\delta)^{\theta (1+\eps_0) (1+\delta)/(1-(1/C))} \Phi(x\su{t}). \]
\end{lemma}
\begin{proof}
$x^*$ is $k$-sparse, so we may assume w.l.o.g., that $x^* = \theta_1 e_1 +
\theta_2 e_2 + \dots + \theta_k e_k$.  Let us define
\[ \Delta_i = \Phi( x\su{t} + \theta_i e_i ) - \Phi( x\su{t} ). \]

First, we will show that
\begin{equation}\label{eq:toshow2}
\sum_{i=1}^k \Delta_i \le \Phi(x\su{t}) \big[ (1+\delta)^{(1+\eps_0)} -1 \big].
\end{equation}

To see this, note that the LHS of Eq.\eqref{eq:toshow2} equals
\begin{align*}
&\sum_{j} b_j \left( \sum_{i=1}^k (1+\delta)^{(A(x\su{t}+\theta_i e_i))_j/b_j} - (1+\delta)^{(Ax\su{t})_j/b_j} \right) \\
&= \sum_j b_j (1+\delta)^{(Ax\su{t})_j/b_j} \left( \sum_{i=1}^k (1+\delta)^{(A(\theta_i e_i))_j/b_j} - 1 \right) \\
&\le \sum_j b_j (1+\delta)^{(Ax\su{t})_j/b_j} \left( (1+\delta)^{(A x^*)_j/b_j} -1\right).
\end{align*}
In the last step, we used the fact that the function $f(t) := (1+\delta)^t -1$ is sub-additive (Lemma~\ref{lem:subadditive}), and the fact that $x^* = \sum_{i=1}^k \theta_i e_i$. Now, using the bound we have on $\Phi(x^*)$, we obtain Eq.~\eqref{eq:toshow2}. The second observation we make is that since $\norm{x^*}_1 = 1$, we have 
$\sum_i \theta_i = 1$.

Now we can apply the averaging lemma~\ref{lem:averaging} with the numbers $\{\Delta_i, \theta_i\}_{i=1}^k$, to conclude that for any $C>1$, there exists an $i \in [k]$ such that $\theta_i \ge 1/(Ck)$, and
\[ \Delta_i \le \Phi(x\su{t}) \cdot \frac{ (1+\delta)^{(1+\eps_0)}-1 }{1- (1/C)}. \]
Thus we have that for this choice of $i$, and $\theta = \theta_i$,
\[ \Phi(x\su{t} + \theta e_i) \le \Phi(x\su{t}) \left( 1+ \theta \cdot \frac{ (1+\delta)^{(1+\eps_0)}-1 }{1- (1/C)} \right). \]
Now we can simplify the term in the parenthesis using Lemma~\ref{lem:basic-ineq} (twice) to obtain
\begin{align*}
1 + \theta &\cdot \frac{ (1+\delta)^{(1+\eps_0)}-1 }{1- (1/C)} \le 1 + \frac{ \theta \cdot \delta (1+\eps_0)}{1-(1/C)} \\
&\le (1+\delta)^{\theta (1+\eps_0) (1+\delta)/(1-(1/C))}.
\end{align*}

This completes the proof of the lemma.
\end{proof}

\begin{proof}[Proof of Theorem~\ref{thm:main-approx}]
By hypothesis, we know that there exists an $\xopt$ such that $\norm{A\xopt}_1=1$ (or equivalently $\norm{\xopt}_1=1$) and $A\xopt \le (1+\xopt) b$. Thus for this $\xopt$, we have $\Phi(\xopt) \le (1+\delta)^{(1+\eps_0)}$, so
Lemma~\ref{lem:progress} shows that in each iteration, the algorithm succeeds in finding an index $i$ and $\theta > 1/Ck$ satisfying the conclusion of the lemma. Thus after $T$ steps, we end up with $\psi(x\su{T}) \ge 1/\delta^2$, thus we can appeal to Lemma~\ref{lem:l1-bound}. Setting $\eta := 2(\eps_0 + \delta + 1/C)$, and observing that
\[ (1+\eps_0)(1+\delta)/(1-1/C) < 1+ \eta, \] the lemma implies that the $\ell_1$ error is
$\le 2(\eta + \delta) < \eps$, from our choice of $\eta, \delta$.
This completes the proof.
\end{proof}

\begin{remark}\label{remark:oracle}
The algorithm above finds a column to add by going through indices $i \in [m]$, and checking if there is a scaling of $A_i$ that can be added. But in fact, any procedure that allows us to find a column with a small value of $\Phi(x\su{t+1})/\Phi(x\su{t})$ would suffice for the algorithm. For example, the columns could be parametrized by a continuous variable, and we may have a procedure that only searches over a discretization.\footnote{This is a fact we will use in our result on learning mixtures of Gaussians.} We could also have an optimization algorithm that outputs the column to add.
\end{remark}


\section{Learning Gaussian mixtures}\label{sec:gaussians}

\subsection{Notation}

Let $N(\mup,\sigp^2)$ denote the density of a $d$-dimensional axis-aligned Gaussian distribution with mean $\mup$ 
and diagonal covariance matrix $\sigp^2$ respectively.  
Thus a $k$-component Gaussian mixture has the density $\sum^k_{r=1} \wi N(\mupi,\sigpi^2)$. 
We use $\bfvec$ to denote the underlying mixture and $\bp_r$ to denote component $r$.
We use ($~\hat{}~$) to denote empirical or other estimates; the usage becomes clear in context.
For an interval $I$, let $|I|$ denote its length.
For a set $S$, let $n(S)$ be the number of samples in that set.

\subsection{Algorithm}
\label{sec:algorithm}
The problem for finding components of a $k$-component Gaussian mixture $f$
 can be viewed as finding a sparse solution for system of equations
\begin{equation}
\label{eq:gauss}
Aw = \bfvec,
\end{equation}
where columns of $A$ are the possible mixture 
components and $w$ is the weight vector
and $\bfvec$ is the density of the underlying mixture.
If $\bfvec$ is known exactly, and $A$ is known explicitly, \eqref{eq:gauss} can be solved using
\textsf{solve}$(\{A, \bfvec, k, \epsilon\})$.

However, a direct application of \textsf{solve}
has two main issues. Firstly, $\bfvec$ takes values 
over $\mathbb{R}^d$ and thus is an infinite dimensional vector.
Thus a direct application of \textsf{solve} is not computationally feasible.
Secondly $\bfvec$ is unknown and has to be estimated  
using samples. Also, for algorithm \textsf{solve}'s performance guarantees to hold,
we need an estimate $\hat{\bfvec}$ such that 
$\hat{\bfvec}(x) \geq \bfvec(x)(1-\epsilon)$, for all $x$.
This kind of a global multiplicative condition is difficult to satisfy 
for continuous distributions. To avoid these issues,  we carefully discretize the mixture of Gaussians. 
More specifically, we partition $\mathbb{R}^d$ into rectangular regions $\cS = \{S_1,S_2,\ldots \}$
such that $S_i \cap S_j = \emptyset$ and $\cup_{S \in \cS} S = \mathbb{R}^d$.
Furthermore we flatten the Gaussian within each region to induce a new distribution over $\rr^d$ as follows:

\begin{defn}
For a distribution $\bp$ and a partition $\cS$, the new distribution $p^{\cS}$ is defined as\footnote{We are slightly abusing notation, with $p^{\cS}$ denoting both the p.d.f. and the distribution itself.} 
\begin{itemize}
\item
If $\bx, \by \in S$ for some $S \in \cS$, then $\bp^{\cS}(\bx) = \bp^{\cS}(\by)$
\item
$\forall S \in \cS$, $\bp(S) = \bp^{\cS}(S)$.
\end{itemize}
\end{defn}

Note that we use the standard notation that $p(S)$ denotes the total probability mass of the distribution $p$ over the region $S$.
Now, let $A^{\cS}$ be a matrix with rows indexed by $S \in \cS$ and columns indexed by distributions $p$
such that $A^{\cS}(S,p) = p(S)$.
 $A^{\cS}$ is a matrix with potentially infinitely many columns, but finitely many rows (number of regions in our partition).

Using samples, we generate a partition of 
$\mathbb{R}^d$ such that the following properties hold.
\begin{enumerate}
\item
${\bfvec}^{\cS}(S)$ can be estimated to sufficient multiplicative 
accuracy for each set $S \in \cS$.
\item
If we output a mixture of $\cO(k/\epsilon^3)$ Gaussians $A^{\cS}w'$ such that
$\sum_{S \in \cS} |(A^{\cS}w')(S) - {\bfvec}^{\cS}(S)|$ is small,
 then $\norm{Aw' - {\bfvec}}_1$ is also small.
\end{enumerate}

For the first one to hold, we require the sets to have large probabilities and hence requires $\cS$ to be a coarse partition of $\mathbb{R}^d$.
The second condition requires the partition to be `fine enough', that a solution after partitioning can be used to produce a solution for the corresponding continuous distributions. How do we construct such a partition?

If all the Gaussian components have similar variances and the means are not too far apart, then a rectangular grid with  carefully chosen width would suffice for this purpose. However, since we make no assumptions
on the variances, we use a sample-dependent partition (i.e., use some of the samples from the mixture in order to get a rough estimate for the `location' of the probability mass). To formalize this,  we need a few more definitions.
\begin{defn}
A partition of a real line is given by $\cI = \{I_1,I_2,\ldots \}$ 
where $I_t$s are continuous intervals, 
$I_t \cap I_{t'} = \emptyset \, \forall t,t'$,
and $\cup_{I \in \cI} I = \mathbb{R}$.
\end{defn}

Since we have $d$ dimensions, we have $d$ such partitions. We denote by $\cI_i$ the partition of axis $i$.  The interval $t$ of coordinate $i$ is denoted by $I_{i,t}$.

For ease of notation, we use subscript $r$ to denote components (of the mixture),
$i$ to denote coordinates ($1 \le i \le d$), and $t$ to denote the interval indices corresponding to coordinates.
We now define induced partition based on intervals
and a notion of ``good'' distributions.

\begin{defn}
Given partitions 
$\cI_{1}, \cI_{2}, \cI_{3}, \ldots \cI_{d}$ for coordinates $1$ to $d$,
define $\cI_{1}, \cI_{2}, \cI_{3}, \ldots \cI_{d}$-induced partition $\cS = \{S_v\}$
as follows: for every $d$-tuple $\bv$, $x \in S_{\bv}$
iff $x_i \in I_{i,v_i} \, \forall \bv$.	
\end{defn}

\begin{defn}
A product distribution $p = p_1 \times p_2 \times \ldots \times p_d$
is $(\cI_{1}, \cI_{2}, \cI_{3}, \ldots \cI_{d}), \epsilon$-good if for every 
coordinate $i$ and every interval $I_{i,t}$, $p_i(I_{i,t}) \leq \epsilon$.
\end{defn}

Intuitively, $\epsilon$-good distributions have small mass in every interval
and hence binning it would not change the distribution by much.
Specifically in Lemma~\ref{lem:ddimensionalbound}, we show that 
for such distributions $\norm{p-p^{\cS}}_1$ is bounded.

We now have all the tools to describe the algorithm.
Let $\epsilon_1 = \epsilon^3/kd$.
The algorithm first divides $\mathbb{R}^d$ into a rectangular 
gridded fine partition $\cS$ 
with $\approx \epsilon^{-d}_1$
 bins such that most of them have probability 
$\geq \epsilon^{d+1}_1$. 
We then group the bins with probability $< \epsilon^{d+1}_1$
to create a slightly coarser partition $\cS'$. 
The resulting $\cS'$ is coarse enough that 
$\bfvec^{\cS'}$ can be estimated efficiently, and is also fine enough to ensure that we do not lose much of the Gaussian structure by binning.

We then limit the columns of
$A^{\cS'}$ to contain only Gaussians
that are $(\cI_1,\cI_2,\ldots \cI_d), 2\epsilon^2/d$-good.
In Lemma~\ref{lem:ddimensionalbound}, we show that 
for all of these we do not lose much of the Gaussian structure by binning.
Thus \textsf{solve}$(A^{\cS'}w, b, k,\epsilon)$ yields us the required solution.
With these definitions in mind, the algorithm is given in \textsf{Learn}$(\{(x_1,\ldots x_{2n}), k, \eps\})$.
Note that the number of rows in $A^{\cS'}$ is $|\cS'| \leq |\cS| = \epsilon^{-d}_1$.

We need to bound the time complexity of finding a Gaussian in each iteration of the algorithm (to apply Remark~\ref{remark:oracle}). 
To this end we need to find a finite set of candidate Gaussians (columns of $A^{\cS'}$) such that running \textsf{solve} 
using a matrix restricted to these columns (call it $A^{\cS'}_\text{finite}$) finds the desired mixture up to error $\epsilon$.
Note that for this, we need to ensure that there is at least one candidate (column of $A^{\cS'}_\text{finite}$) that is close to each of the true mixture components.

We ensure this as follows.
Obtain a set of $n'$ samples from the Gaussian mixture and for each pair of samples $x,y$ consider the Gaussian whose mean is $x$
and the variance along coordinate $i$ is $(x_i-y_i)^2$. Similar to the proof of the 
one-dimensional version in~\cite{AcharyaJOS14}, 
it follows that for any $\epsilon'$ choosing 
$n' \geq \Omega((\epsilon')^{-d})$, this set contains Gaussians that are $\epsilon'$ close to each of the underlying mixture components. 
For clarity of exposition, we ignore this additional error which can be made arbitrarily small and we treat $\epsilon'$ as $0$.



\begin{algorithm}[!ht]
  ~~procedure \textsf{Learn}($\{(x_1,\ldots x_{2n}), k, \eps\}$) \\
\Begin{

      \nl Set parameter $\epsilon_1 = \epsilon^3/(kd)$.

      \nl Use first $n$ samples to find 
      $\cI_1,\cI_2 \ldots \cI_{d}$ such that
      number of samples $\bx$ such that 
      $x_i \in  I_{i,t}$ is $n\epsilon_1$. 
      Let $\cS$ be the corresponding induced partition.

      \nl Use the remaining $n$ samples to do:

      \nl \, \, Let $U = \cup S_{\bv} : n(S_{\bv}) \leq n \epsilon^{d}_1 \epsilon$.

\nl  \, \, Let $\cS' = \{U\} \cup \{S \in \cS : n(S_{\bv}) > n \epsilon^{d}_1 \epsilon \}$.

\nl  \, \, Set $b(U) = 2\epsilon$ and \[ \forall S \in \cS' \setminus \{U\},
b(S) = \frac{(1-2\epsilon) n(S)}{ \sum_{S \in \cS' \setminus \{U\}} n(S)}\]

      \nl 
      Let $A^{\cS'}$ be the matrix with columns corresponding to
      distributions $p$ that are  $(\cI_1,\cI_2,\ldots \cI_d), 2\epsilon^2/d$-good axis-aligned
      Gaussians, and $A^{\cS'}_\text{finite}$ be the candidates obtained as above, using $\epsilon' = \epsilon_1/10$.

\nl
\textsf{solve}($A^{\cS'}_\text{finite}, b, k, 64\epsilon)$ using Remark~\ref{remark:oracle}.

\nl
Output the $w$.

}
\end{algorithm}

\subsection{Proof of correctness}

We first show that $b$ satisfies the necessary conditions for \textsf{solve} that
are given in Theorem~\ref{thm:main-approx}. The proof follows from Chernoff bound 
and the fact that empirical mass in most sets $S \in \cS'$ is $\geq \epsilon^{d}_1\epsilon$.
\begin{lemma}
\label{lem:fisclosetob}
If $n \geq  \frac{8}{\epsilon^{d}_1 \epsilon^{3}} \log \frac {2}{\delta \epsilon^{d}_1 }$, then with probability $\geq 1-\delta$
\[
\forall S \in \cS' ,b(S) \geq \bfvec^{\cS'}(S)(1-3\epsilon),
\]
and $\sum_{S \in \cS'}|\bfvec^{\cS'}(S)-b(S)| \leq 6 \epsilon$.
%
\end{lemma}
\begin{proof}
Let $\hat{\bfvec}^{\cS'}$ be the empirical distribution over $\cS'$.
Since $|\cI_i| = \frac{1}{\epsilon_1}$,
the induced partition $\cS$ satisfies 
$|\cS| \leq \frac{1}{\epsilon^{d}_1}$. 
Hence by the Chernoff and union bounds, for
$n \geq \frac{8}{\epsilon^{d}_1 \epsilon^{3}} \log \frac {2}{\delta \epsilon^{d}_1 }$, with probability $\geq 1 -\delta$,
\begin{equation}
\label{eq:chern}
|\bfvec^{\cS'}(S) - \hat{\bfvec}^{\cS'}(S)| \leq \sqrt{\bfvec^{\cS'}(S)\epsilon^{d}_1\epsilon^3}/2 + \epsilon^{d}_1\epsilon^{3}/2, \, \forall S \in \cS.
\end{equation}
For the set $U$,
\begin{align*}
\bfvec^{\cS'}(U) 
&= \sum_{S : \hat{\bfvec}^{\cS'}(S) \leq \epsilon^{d}\epsilon } \bfvec^{\cS'}(S) \\
& \leq \epsilon + \sum_{S \in \cS} \sqrt{\bfvec^{\cS'}(S) \epsilon^{d}_1 \epsilon^3 }/2 + \sum_{S \in \cS} \epsilon^{d}_1\epsilon^{3}/2\\
& \leq 2\epsilon,
\end{align*}
where the second inequality follows from the concavity of $\sqrt{x}$. However $b(U) = 2 \epsilon$ and hence $b(U) \geq \bfvec^{\cS'}(U)(1-2\epsilon)$.

By Equation~\eqref{eq:chern}, 
\begin{align*}
|\bfvec^{\cS'}(S) - \hat{\bfvec}^{\cS'}(S)| 
& \leq \sqrt{\bfvec^{\cS'}(S)\epsilon^{d}_1 \epsilon^3}/2 + \epsilon^{d}_1 \epsilon^3/2  \\
& \leq  \sqrt{\hat{\bfvec}^{\cS'}(S)\epsilon^{d}_1 \epsilon^3}/2 + \epsilon^{d}_1 \epsilon^3/2 \\
& \leq \hat{\bfvec}^{\cS'}(S) \left( \sqrt{ \epsilon^2}/2 + \epsilon^2/2  \right) \\
& \leq \hat{\bfvec}^{\cS'}(S) \epsilon.
\end{align*}
The penultimate inequality follows from the fact that $\hat{\bfvec}^{\cS'}(S) \geq \epsilon^{d}_1 \epsilon$.
Hence $\hat{\bfvec}^{\cS'}(S) \geq \bfvec^{\cS'}(S)(1-\epsilon)$.
Furthermore by construction $b(S) \geq \hat{\bfvec}^{\cS'}(S)(1-2\epsilon)$. Hence $b(S) \geq \bfvec^{\cS'}(S)(1-3\epsilon) \, \forall S \in \cS'$.

For the second part of the lemma observe that $b$ and $\bfvec^{\cS'}$ are distributions over $\cS'$. Hence
\begin{align*}
\sum_{S \in \cS} |b(S) -\bfvec^{\cS'}(S)| 
& = 2 \sum_{S : b(S) \leq \bfvec^{\cS'}(S)} 
 \bfvec^{\cS'}(S) -b(S) \\
& \leq 2 \sum_{S \in S} \bfvec^{\cS'}(S) \cdot 3 \epsilon = 6\epsilon.
\end{align*}
\end{proof}

Using the above lemma, we now prove that \textsf{Learn} returns a good solution such that 
$\norm{\bfvec^{\cS} - \hat{\bfvec}^{\cS}}_1 \leq \cO(\epsilon)$.

\begin{lemma}
\label{lem:solutionovers}
Let $n \geq \max \left( \frac{2}{\epsilon^2_1} \log \frac{2d}{\delta},
 \frac{8}{\epsilon^{d}_1 \epsilon^{3}} \log \frac {2}{\delta \epsilon^{d}_1} \right) $. With probability $\geq 1 -2\delta$, 
\textsf{Learn} returns a solution $\hat{\bfvec}$ such that 
the resulting mixture satisfies
\[
\norm{\bfvec^{\cS} - \hat{\bfvec}^{\cS}}_1 \leq 74 \epsilon.
\]
\end{lemma}
\begin{proof}
We first show that $A^{\cS'}$
has columns corresponding to  all the components $r$, 
such that $\wi \geq \epsilon/k$. 
For a mixture $\bfvec$ let $f_i$ be the projection of $\bfvec$
on coordinate $i$.
Note that $\hat{f}_i(I_{i,t}) = \epsilon_1 \, \forall i,t$.
Therefore by Dvoretzky-Kiefer-Wolfowitz theorem (see, ~\cite{Massart90}) and the union bound if $n \geq \frac{2}{\epsilon^2_1} \log \frac{2d}{\delta}$, with probability $\geq 1- \delta$,
\[
f_{i}(I_{i,t}) \leq \epsilon_1 + \epsilon_1 \leq 2 \epsilon_1 \, \forall i,t.
\]
Since $f_{i} = \sum^k_{r=1} \wj p_{r,i}$, with probability $\geq 1- \delta$,
\[
p_{r,i} (I_{i,t}) \leq \frac{2\epsilon_1}{\wj} \, \forall i,r.
\]
If $\wj \geq \epsilon/k$, then $p_{r,i}(I_{i,t}) \leq 2\epsilon^2/d$
and thus $A^{\cS'}$ contains all the underlying 
components $r$ such that $\wj \geq \epsilon/k$.
Let $w^{*}$ be the weights corresponding to components such that $\wj \geq \epsilon/k$.
Therefore $\norm{w^{*}}_1 \geq 1-\epsilon$. 
Furthermore by Lemma~\ref{lem:fisclosetob},
 $b(S) \geq f^{\cS'}(1-3\epsilon) \geq (1-3\epsilon)(A^{\cS'}w^{*})(S)$.
 Therefore, we have 
 $\norm{b} = \norm{A^{\cS'}w^*/\norm{w^*}}$
 and 
 \begin{align*}
  b(S) 
  & \geq (1-3\epsilon)(A^{\cS'}w^{*})(S) \norm{w^{*}}/\norm{w^{*}} \\
  & \geq (1-4\epsilon) (A^{\cS'}w^{*}/\norm{w^{*}})(S).
 \end{align*}
Hence, By Theorem~\ref{thm:main-approx}, 
algorithm returns a solution $A^{\cS'}w' $
such that $\norm{A^{\cS'}w' - b}_1 \leq 64 \epsilon$.
Thus by Lemma~\ref{lem:fisclosetob}, 
$\sum_{S \in \cS'} |(A^{\cS'}w')(S) - f^{\cS'}(S)| \leq 70 \epsilon$.
Let $\hat{\bfvec}$ be the estimate corresponding to solution $w'$.
Since $f^{\cS'}$ and $\hat{f}^{\cS'}$ are flat within sets $S$,
we have $\norm{\hat{f}^{\cS'} - f^{\cS'}}_1 \leq 70 \epsilon$.

Since $\cS'$ and $\cS$ differ only in the 
set $U$ and by Lemma~\ref{lem:fisclosetob}, $\bfvec^{\cS}(U) = \bfvec^{\cS'}(U) \leq \epsilon/(1-3\epsilon)$,
we have 
\[
\norm{\hat{\bfvec}^{\cS} - f^{\cS}}_1 \leq \norm{\hat{\bfvec}^{\cS'} - f^{\cS'}}_1 + 
2\bfvec^{\cS}(U) \leq 74 \epsilon.
\]
Note that the total error probability is $\leq 2\delta$.
\end{proof}

We now prove that if $\hat{\bfvec}^{\cS}$ is close to $\bfvec^{\cS}$,
then $\hat{\bfvec}$ is close to $\bfvec$. We first prove that 
flattened Gaussians in one dimension are
close to their corresponding underlying Gaussian.

\begin{lemma}
\label{lem:onedimensionalbound}
Let $p$ be a one dimensional Gaussian distribution 
and $\cI = (I_1,I_2,\ldots)$ be a partition of the real line such that $\forall I \in \cI$,
 $I$ is a continuous interval
and $p(I) \leq \epsilon$. Then
\[
\norm{p - p^{\cI}}_1 \leq 30 \sqrt{\epsilon}.
\]
\end{lemma}
\begin{proof}
If $p$ and $\cI$ are simultaneously scaled or translated,
then the value of $\norm{p-p^{\cI}}_1$ remains unchanged. 
Hence proving the lemma for $p = N(0,1)$ is sufficient.
We first divide $\cI$ into $\cI_1,\cI_2,\cI_3$ 
depending on the minimum and maximum values of $p(x)$ in the corresponding intervals.
\[
I 
\in 
\begin{cases}
\cI_{1} & \text{\, if\, } \min_{x \in I} p(x) \geq \sqrt{\epsilon/(2\pi)}, \\
\cI_{2} & \text{\, if\, } \max_{x \in I} p(x) \leq \sqrt{\epsilon/(2\pi)}, \\
\cI_{3}  & \text{\, else.\,}
\end{cases}
\]
The $\ell_1$ distance between $p$ and $p^{\cI}$ is
\[
\norm{p-p^{\cI}}_1 = \sum_{I \in \cI} \int_{x \in I} | p(x) - p^{\cI}(x)|dx.
\]
We bound the above summation by breaking it into
terms corresponding to $\cI_1, \cI_2$, and $\cI_3$
respectively.
Observe that $|\cI_{3}| \leq 2$ and $p(I) \leq \epsilon\, \forall \, I \in \cI_{3}$. Hence,
\[
 \sum_{I \in \cI_3} \int_{x \in I} | p(x) - p^{\cI}(x)|dx \leq 2 \epsilon.
\]
Since $\max_{x \in I} p(x)$ for every interval in $I \in \cI_{2}$ is $\leq \sqrt{\epsilon/(2\pi)}$, by Gaussian tail bounds
\begin{align*}
 \sum_{I \in \cI_2} \int_{x \in I} | p(x) - p^{\cI}(x)|dx
 & \leq  \sum_{I \in \cI_2} \int_{x \in I}  p(x) dx \\
& \leq \sqrt{\epsilon}.
\end{align*}
For every interval $I \in \cI_{1}$ we first bound its interval length and maximum value of $p'(x)$.
Note that
\[
p(I) \geq |I| \min_{y \in I} p(y).
\]
In particular since $p(I) \leq \epsilon$ and $\min_{y \in I} p(y) \geq \sqrt{\epsilon/(2\pi)}$, 
$|I| \leq \sqrt{2\pi\epsilon}$.
Let $s = \max_{x \in I} |p'(x)|$.
\[
s =  \max_{x \in I} |p'(x)|
 = \max_{x \in I}   \frac{\lvert x \rvert }{\sqrt{2\pi}} e^{-x^{2}/2}
\leq  \max_{x \in I} |x| \cdot  \max_{x \in I} p(x).
\]
Since $\min_{y \in I} p(y) \geq \sqrt{\epsilon/(2\pi)}$, we have 
$\max_{y \in I} |y| \leq \sqrt{\log 1/\epsilon}$.
Let $y_1 = \text{argmax}_{y \in I} p(y)$ and $y_2 = \text{argmin}_{y \in I} p(y)$,
then 
\begin{align*}
\frac{\max_{y \in I} p(y)}{\min_{y \in I} p(y)} 
& = \frac{p(y_2)}{p(y_{1})}  \\
& = e^{(y^2_1 - y^2_2)/2} \\
&= e^{(y_1 - y_2)(y_2+y_1)/2} \\
& \leq e^{|I| \sqrt{\log \frac{1}{\epsilon}}} \\
& \leq e^{\sqrt{2\pi \epsilon \log \frac{1}{\epsilon}}}\\
& \leq 5.
\end{align*}
Since $p^{\cI}(x) = p(I)/|I|$, by Rolle's theorem
$\exists x_0$ such that $p^{\cI}(x) = p(x_0) \, \forall x$.
By first order Taylor's expansion,
\begin{align*}
\int_{x \in I} | p(x) - p^{\cI}(x)|dx 
& \leq \int_{x \in I} |(x-x_0) \max_{y \in [x_0,x]} |p'(y)|dx \\
& \leq s \int_{x \in I} |x -x_0|   dx \\
& \leq s |I|^2 \\
& \leq s \left( \frac{p(I)}{\min_{y \in I} p(y)}\right) \sqrt{2 \pi \epsilon} \\
& \leq  \sqrt{2\pi\epsilon} p(I) \max_{x \in I} |x| \cdot \frac{\max_{y \in I} p(y)}{\min_{y \in I} p(y)} \\
& \leq 5 \sqrt{2\pi\epsilon} p(I) \max_{x \in I} |x|,
\end{align*}
where the last three inequalities follow from the bounds on 
$|I|$, $s$, and $\frac{\max_{y \in I} p(y)}{\min_{y \in I} p(y)}$
respectively.
Thus,
\begin{align*}
\int_{x \in I} | p(x) - p^{\cI}(x)|dx  
& \leq 5 \sqrt{2\pi\epsilon} p(I) \max_{x \in I} |x| \\
& \leq 5 \sqrt{2\pi\epsilon} \int_{x \in I} p(x) (|x| + \sqrt{\epsilon}) dx.
\end{align*}
Summing over $I \in \cI_{1}$, we get the above summation is 
$\leq 5 \sqrt{2\pi\epsilon}(1+\sqrt{\epsilon})$.
Adding terms corresponding to $\cI_1, \cI_2$, and $\cI_3$ we get
\[
\norm{p-p^{\cI}}_1  \leq 5 \sqrt{2\pi\epsilon} (1 + \sqrt{\epsilon}) + \sqrt{\epsilon} + 2 \epsilon < 30 \sqrt{\epsilon}.
\]
\end{proof}

Using the above lemma we now show that for every $d$-dimensional
 $(\cI_1,\cI_2\ldots \cI_d),\epsilon$-flat Gaussian is close to the unflattened one.

\begin{lemma}
\label{lem:ddimensionalbound}
For every $(\cI_1,\cI_2,\ldots \cI_d),\epsilon$-good axis-aligned Gaussian distribution $\bp = p_1\times p_2\times \ldots p_d$, we have 
\[
\norm{\bp- \bp^{\cS}}_1 \leq 30d \sqrt{\epsilon}.
\]
\end{lemma}
\begin{proof}
By triangle inequality, the distance between any two product distributions
is upper bounded by the sum of distances in each coordinate. Hence,
\[
\norm{\bp- \bp^{\cS}}_1 \leq \sum^d_{i=1} \norm{p_i- p_i^{\cI_{i}}}_1 \leq 30d \sqrt{\epsilon},
\]
where the second inequality follows from Lemma~\ref{lem:onedimensionalbound}.
\end{proof}
We now have all the tools to prove the main result on Gaussian mixtures.
\begin{theorem}
\label{thm:learngaussian}
Let $\epsilon_1 = \epsilon^3/kd$ and 
 $n \geq \max \left( \frac{2}{\epsilon^2_1} \log \frac{2d}{\delta},
 \frac{8}{\epsilon^{d}_1 \epsilon^{3}} \log \frac {2}{\delta \epsilon^{d}_1} \right) $. 
 Then given $2n$ samples from an axis-aligned Gaussian mixture $\bfvec$,
with probability $\geq 1 - 2\delta$, \textsf{Learn}
returns an estimate mixture $\hat{\bfvec}$ with at most $\cO(k/\epsilon^3)$ components such that  
\[
\norm{\hat{\bfvec} - \bfvec}_1 \leq 170\epsilon.
\]
The run time of the algorithm is $\cO\left(1/\epsilon_1 \right)^d$.
\end{theorem}
\begin{proof}
By triangle inequality,
\[
\norm{\hat{\bfvec} - \bfvec}_1 \leq
\norm{\hat{\bfvec}^{\cS} - \bfvec^{\cS}}_1 + \norm{\hat{\bfvec}^{\cS} - \hat{\bfvec}}_1 + \norm{\bfvec^{\cS} - \bfvec}_1.
\]
We now bound each of the terms above.
By Lemma~\ref{lem:solutionovers}, the first term is $\leq 74\epsilon$.
By triangle inequality for $\hat{\bfvec} = \sum^{k'}_{r=1} \hat{\wi} \hat{\bp}_r$,
\[
\norm{\hat{\bfvec}^{\cS} - \hat{\bfvec}}_1
\leq \sum^{k'}_{r =1} \hat{\wi} \norm{\hat{\bp}^{\cS}_r - \hat{\bp}_r}_1 \leq 30\sqrt{2} \epsilon,
\]
where the last inequality follows from the fact that the allowed distributions in $A^{\cS'}
$ are $(\cI_1,\cI_2,\ldots \cI_d),2\epsilon^2/d$-good and by 
Lemma~\ref{lem:ddimensionalbound}.
By triangle inequality,
\begin{align*}
\norm{\bfvec^{\cS} - \bfvec}_1
& \leq \sum^{k}_{r =1} \wi \norm{\bp^{\cS}_r - \bp_r}_1  \\
& \leq \!  \!\sum_{r: \wi \geq \epsilon/k} \wi \norm{\bp^{\cS}_r - \bp_r}_1 \!  \!
+ \!  \! \sum_{r: \wi < \epsilon/k} \wi \norm{\bp^{\cS}_r - \bp_r}_1 \\
& \leq \sum_{r: \wi \geq \epsilon/k} \wi \norm{\bp^{\cS}_r - \bp_r}_1 
+ 2\epsilon \\
& \leq 30\sqrt{2}\epsilon + 2\epsilon.
\end{align*}
where the last inequality follows from the proof of Lemma~\ref{lem:solutionovers}, where we showed that heavy components are $(\cI_1,\cI_2,\ldots \cI_d),2\epsilon^2/d$-good and by Lemma~\ref{lem:ddimensionalbound}.
Summing over the terms corresponding to $\norm{\hat{\bfvec}^{\cS} - \bfvec^{\cS}}_1$, $\norm{\hat{\bfvec}^{\cS} - \hat{\bfvec}}_1$, and  $\norm{\bfvec^{\cS} - \bfvec}_1$, we get the total error as 
$74\epsilon+ 30\sqrt{2} \epsilon + 30 \sqrt{2}\epsilon+2\epsilon \leq 170\epsilon$.
The error probability and the number of samples necessary are same as that of Lemma~\ref{lem:solutionovers}.
The run time follows from the comments in Section~\ref{sec:algorithm} and the bound on number of samples.
\end{proof}

If we consider the leading term in sample complexity, for $d =1$ our bound is $\tcO(k^2/\epsilon^6)$, and for $d > 1$, our bound is  $\tcO((kd)^{d}/\epsilon^{3d+3})$. While this is not the optimal sample complexity (see~\cite{AcharyaJOS14}),  we gain significantly in the running time.


\section{Lower bounds} \label{sec:lower-bounds}
We now investigate lower bounds towards obtaining sparse approximate solutions to nonnegative systems. 
Our first result is that unless $\mathcal{P} = \mathcal{NP}$, we need to lose a factor at least $\log(1/\eps)$ in the sparsity to be $\eps$-close in the $\ell_1$ norm. 
Formally, 

\begin{theorem}\label{Theorem:max_k_cover}
For any $\epsilon > 0$, given an instance of the sparse recovery problem $A, b$ that is promised to have a $k$-sparse nonnegative solution, it is $\mathcal{NP}$-hard to obtain an $o\left(k\ln\left(\frac{1}{\epsilon}\right)\right)$-sparse solution $\xalg$ with $\norm{A\xalg - b}_1 < \eps \norm{b}_1$.
\end{theorem}

Our second result gives a connection to a random planted
version of the set cover problem, which is beyond all
known algorithmic techniques. We
prove that unless this planted set cover problem can
be solved efficiently, we cannot hope to obtain an $\epsilon$-
approximate solution with sparsity $o(k/\epsilon^2)$.

Theorem~\ref{Theorem:max_k_cover} is 
inspired by the hard instances of Max $k$-Cover problem
~\citep{Fei98,FT04,FV10}. 

\paragraph{Hard Instances of Max $k$-Cover.} 
For any $c > 0$, and $\delta > 0$, given a collection of $n$ sets $S_1,
S_2, \dots, S_n \subseteq [m]$, it is $\mathcal{NP}$-Hard to distinguish  
between the following two cases: 

\begin{itemize}
\item 
{\sc Yes} case: There are $k$ disjoint sets in this collection whose union is $[m]$. 
\item 
{\sc No} case: The union of any $\ell \leq ck$ sets of this collection has size at most $(1-(1-\frac{1}{k})^{\ell} + \delta)n$. 
\end{itemize}
\begin{proof}[Proof outline, Theorem~\ref{Theorem:max_k_cover}]
We reduce hard instance of the Max $k$-cover problem to our problem as follows. For each set $S_i$, we set $A_i$ (column $i$ in $A$) to be the indicator vector of set $S_i$. We also let $b$ to be the vector with all entries equal to one. 

In the {\sc Yes} case, we know there are $k$ disjoint sets whose union is the universe, and we construct solution $x^*$ as follows. We set $x^*_i$ (the $i$th entry of $x^*$) to one if set $S_i$ is one of these $k$ sets, and zero otherwise. It is clear that $Ax^*$ is equal to $b$, and therefore there exists a $k$-sparse solution in the {\sc Yes} case. 

On the other hand, for every $\epsilon$-approximate non-negative solution $\hat{x}$, we know that the number of non-zero entries of $A\hat{x}$ is at most $\epsilon m$ by definition. Define $C$ to be the sub-collection of sets with non-zero entry in $\hat{x}$, i.e. $\{S_i ~|~ \hat{x}_i > 0\}$. We know that each non-zero entry in $A\hat{x}$ is covered by some set in sub-collection $C$. In other words, the union of sets in $C$ has size at least $(1-\epsilon)m$. We imply that the number of sets in collection $C$ should be at least $\Omega(k\ln(\frac{1}{\epsilon+\delta}))$ since $(1-\frac{1}{k})^k$ is in range $[\frac{1}{4}, \frac{1}{e}]$. We can choose $\delta$ to be $\epsilon$, and therefore the sparsest solution that one can find in the {\sc No} case is $\Omega(k\ln(\frac{1}{\epsilon}))$-sparse. 
Assuming $\mathcal{P} \neq \mathcal{NP}$, it is not possible to find a $o(k\ln{\frac{1}{\epsilon}})$-sparse $\epsilon$-approximate solution when there exists a $k$-sparse solution, otherwise it becomes possible to distinguish between the {\sc Yes} and {\sc No} cases of the Max $k$-Cover problem in polynomial time. 
\end{proof}

Finally, we show that unless a certain variant of set cover can be solved efficiently, we cannot hope to obtain an $\eps$-approximate solution with sparsity $o(k/\eps^2)$. We will call this the {\em planted set cover} problem:

\begin{defn}{(Planted set cover $(k,m)$ problem)}
Given parameters $m$ and $k > m^{3/4}$, find an algorithm that distinguishes with probability $> 2/3$ between the following distributions over set systems over $m$ elements and $n = O(m/ \log m)$ sets:

{\sc No} case: The set system is random, with element $i$ in set $j$ with probability $1/k$ (independently).

{\sc Yes} case: We take a random set system with $n-k$ sets as above, and add a random $k$-partition of the elements as the remaining $k$ sets. (Thus there is a perfect cover using $k$ sets.)
\end{defn}

To the best of our knowledge, none of the algorithmic techniques developed in the context of set cover can solve this distinguishing problem. The situation is similar in spirit to the planted clique and planted dense subgraph problems on random graphs, as well as random $3$-SAT
~\citep{PlantedClique, PlantedDks, RandomSat}. This shows that obtaining sparse approximate solutions with sparsity $o(k/\eps^2)$ requires significantly new techniques. Formally, we show the following
\begin{theorem}\label{thm:planted-set-cover}
Let $m^{3/4} < k < m/\log^2 m$. Any algorithm that finds an $o(k/\eps^2)$ sparse $\eps$-approximate solution to non-negative linear systems can solve the planted set cover $(k,m)$ problem.
\end{theorem}
\begin{proof}
Let $n$ (which is $O(m/ \log m)$) be the number of sets in the set system. Let $A$ be the $m \times n$ matrix whose $i,j$th entry is $1$ if element $i$ is in set $j$, and $0$ otherwise. 
It is clear that in the {\sc Yes} case, there exists a solution to $Ax = \ones$ of sparsity $k$. It suffices to show that in the {\sc No} case, there is no $\eps$-approximate solution to $Ax = \ones$ with fewer than $\Omega(k/\eps^2)$ entries.

Let us define $C = 1/\eps^2$, for convenience.  The proof follows the standard template in random matrix theory (e.g.~\cite{Rudelson}): we show that for any {\em fixed} $Ck$-sparse vector $x$, the probability that $\norm{Ax - \ones}_1 < 1/(4\sqrt{C})$ is tiny, and then take a union bound over all $x$ in a fine enough grid to conclude the claim for all $k$-sparse $x$.

Thus let us fix some $Ck$ sparse vector $x$ and consider the quantity $\norm{Ax-\ones}_1$. Let us then consider one row, which we denote by $y$, and consider $|\iprod{y, x} -1|$.  Now each element of $y$ is $1$ with probability $1/k$ and $0$ otherwise (by the way the set system was constructed). Let us define the mean-zero random variable $W_i$, $1\le i \le n$, as follows:
\[ W_i = \begin{cases} 1- 1/k \qquad \text{with probability $1/k$, } \\ -1/k \qquad \text{otherwise.}  \end{cases} \]
We first note that $\E[ |\iprod{y,x}-1|^2] \ge \E[ (\sum_{i} W_i x_i )^2 ]$. This follows simply from the fact that for any random variable $Z$, we have $\E[ |Z-1|^2 ] \ge \E[ |Z-\E[Z]|^2 ]$ (i.e., the best way to ``center'' a distribution with respect to a least squares objective is at its mean). Thus let us consider
\[ \E \left[ \left( \sum_i W_i x_i \right)^2 \right] = \sum_i x_i^2 \cdot \E[W_i^2] = \sum_i x_i^2 \cdot \frac{1}{k} \big( 1- \frac{1}{k} \big). \]

Since $x$ is $Ck$-sparse, and since $\norm{x}_1 \ge 3k/4$, we have $\sum_i x_i^2 \ge \frac{1}{Ck} \cdot \norm{x}_1^2 \ge k/2C$. Plugging this above and combining with our earlier observation, we obtain
\begin{equation}\label{eq:each-coordinate}
\E[ |\iprod{y, x}-1|^2] \ge \E \left[ \left( \sum_i W_i x_i \right)^2 \right] \ge \frac{1}{3C}. 
\end{equation}
Now we will use the Paley-Zygmund inequality,\footnote{For any non-negative random variable $Z$, we have $\Pr( Z \ge \theta \E[Z]) \ge (1-\theta)^2 \cdot \frac{\E[Z]^2}{\E[Z^2]}$.} with the random variable $Z := |\iprod{y, x}-1|^2$. For this we need to upper bound $\E[Z^2] = \E[ |\iprod{y, x} -1 |^4]$. We claim that we can bound it by a constant. Now since $\norm{x}_1$ is between $3k/4$ and $5k/4$, we have $|\iprod{y, x} - \sum_i W_i x_i| < 1/2$. This in turn implies that $\E[Z^2] \le 4( \E[(\sum_i W_i x_i )^4] + 4)$. We will show that $\E[ (\sum_i W_i x_i )^4 ]= O(1)$.
\begin{align*}
\E[ (\sum_i W_i x_i)^4 ] &= \sum_i W_i^4 x_i^4 + 3\sum_{i,j} W_i^2 W_j^2 x_i^2 x_j^2 \\
&\le \frac{1}{k} \cdot \sum_i x_i^4 + \frac{3}{k^2} \sum_{i,j} x_i^2 x_j^2 \\
& \le 1 + \frac{3}{k^2} \cdot (\sum_i x_i^2 )^2 = O(1).
\end{align*}
Here we used the fact that we have $0\le x_i \le 1$ for all $i$, and that $\sum_i x_i^2 \le \sum_i x_i \le 5k/4$.

This implies, by using the Paley-Zygmund inequality, that
\begin{equation}
\Pr \left[ |\iprod{y, x}-1| < \frac{1}{4\sqrt{C}} \right] < 1-1/10.
\end{equation}

Thus if we now look at the $m$ rows of $A$, and consider the number of them that satisfy $|\iprod{y, x}-1| < 1/(4 \sqrt{C})$, the expected number is $< 9m/10$, thus the probability that there are more than $19m/20$ such rows is $\exp(-\Omega(m))$. Thus we have that for any $Ck$-sparse $x$ with $\norm{x}_1 \in [3k/4, 5k/4]$ and $\norm{x}_\infty \le 1$, 
\begin{equation}\label{eq:single-x}
\Pr \left[ \norm{Ax - \ones}_1 < \frac{1}{80\sqrt{C}} \right] < e^{-m/40}.
\end{equation}

Now let us construct an $\eps'$-net\footnote{An $\eps$-net for a set $S$ of points is a subset $T$ such that for any $s \in S$, there is a $t \in T$ such that $\norm{s-t}_2 < \eps$.} for the set of all $Ck$-sparse vectors, with $\eps' = 1/m^2$. A simple way to do it is to first pick the non-zero coordinates, and take all integer multiples of $\eps'/m$ as the coordinates. It is easy to see that this set of points (call it $\calN$) is an $\eps'$ net, and furthermore, it has size roughly
\[ \binom{m}{Ck} \left( \frac{m}{\eps'}  \right)^{Ck} = O\left( m^{4Ck} \right). \]

Thus as long as $m > 200 Ck \log m$, we can take a union bound over all the vectors in the $\eps'$ net, to conclude that with probability $e^{-\Omega(m)}$, we have
\[ \norm{Ax - \ones}_1 > \frac{1}{80 \sqrt{C}} \qquad \text{for all $x \in \calN$.} \]
In the event that this happens, we can use the fact that $\calN$ is an $\eps'$ net (with $\eps' = 1/m^2$), to conclude that $\norm{Ax-\ones}_1 > \frac{1}{100\sqrt{C}}$ for {\em all} $Ck$-sparse vectors with coordinates in $[0,1]$ and $\norm{x}_1 \in [3k/4, 5k/4]$.

This completes the proof of the theorem, since $\frac{1}{100\sqrt{C}}$ is $\Omega(\eps)$, and $k < m/\log^2 m$.
\end{proof}


\bibliographystyle{abbrvnat}
\bibliography{lins}

\newpage 
\appendix
\section{Auxiliary lemmas}
The simple technical lemma we required in the proof is the following.
\begin{lemma}\label{lem:subadditive}
Let $\delta >0$, and $f(t) := (1+\delta)^t-1$. Then for any $t_1, t_2
\ge 0$, we have
\[ f(t_1) + f(t_2) \le f(t_1+t_2). \]
\end{lemma}
\begin{proof}
The proof follows immediately upon expansion:
\[ f(t_1+ t_2) - f(t_1) - f(t_2) = \big( (1+\delta)^{t_1} - 1 \big) \big( (1+\delta)^{t_2} -1 \big).\]
The term above is non-negative because $\delta, t_1, t_2$ are all $\ge
0$.
\end{proof}

\begin{lemma}[Averaging] \label{lem:averaging}
Let $\{a_i, b_i\}_{i=1}^k$ be non-negative real numbers, such that
\[ \sum_i a_i = A \qquad \text{and} \qquad \sum_i b_i =1. \]
Then for any parameter $C>1$, there exists an index $i$ such that $b_i \ge 1/(Ck)$, and $a_i \le b_i \cdot A/(1-1/C)$.
\end{lemma}
\begin{proof}
Let $S := \{ i~:~ b_i \ge 1/(Ck) \}$. Now since there are only $k$ indices, we have $\sum_{i \in [k]\setminus S} b_i < k \cdot 1/(Ck) < 1/C$, and thus
\begin{equation}\label{eq:avging-1}
\sum_{i \in S} b_i > (1-1/C).
\end{equation}
Next, since all the $a_i$ are non-negative, we get that 
\[ \sum_{i \in S} a_i \le A. \]
Combining the two, we have 
\[ \frac{ \sum_{i \in S} a_i }{\sum_{i \in S} b_i } < \frac{A}{1-1/C}. \]
Thus there exists an index $i \in S$ such that $a_i < b_i \cdot A/ (1-1/C)$ (because otherwise, we have $a_i \ge b_i A/(1-1/C)$ for all $i$, thus summing over $i \in S$, we get a contradiction to the above). This proves the lemma.
\end{proof}

\begin{lemma}\label{lem:basic-ineq}
For any $0<x<1$ and $\delta>0$, we have
\[ (1+\delta)^x \le 1+\delta x \le (1+\delta)^{x(1+\delta)}. \]
\end{lemma}
\begin{proof}
For any $0 < \theta < \delta$, we have
\[ \frac{1}{1+\theta} < \frac{1}{1+\theta x} < \frac{1+\delta}{1+\theta}. \]
The first inequality is because $x<1$, and the second is because the RHS is bigger than $1$ while the LHS is smaller. Now integrating from $\theta =0$ to $\theta = \delta$, we get
\[ \log (1+\delta) < \frac{ \log(1+x \delta)}{x} < (1+\delta) \log(1+\delta). \]
Multiplying out by $x$ and exponentiating gives the desired claim.
\end{proof}

\end{document}